\documentclass[12pt]{amsart}
\usepackage{amsfonts,amssymb,amsmath}
\usepackage[pagewise]{lineno}
\usepackage{url}
\usepackage{enumerate}
\usepackage{comment} 
\usepackage{color}
\usepackage[singlelinecheck=false]{caption}

\newtheorem{theorem}{Theorem}[section]

\newtheorem{proposition}[theorem]{Proposition}
\newtheorem{corollary}[theorem]{Corollary}
\theoremstyle{definition}
\newtheorem{definition}[theorem]{Definition}
\newtheorem{remark}[theorem]{Remark}
\numberwithin{equation}{section}

\author{V. R. Bazao}   
\address{Faculdade de Ci\^encias Exatas e Tecnologias, UFGD, Dourados, MS, 79804-970  Brazil}

\author{C. R. de Oliveira} 
\address{Departamento de Matem\'atica, UFSCar, S\~ao Carlos, SP, 13560-970 Brazil}

\author{P. A. Diaz}   
\address{Departamento de Matem\'atica, UFSCar, S\~ao Carlos, SP, 13560-970 Brazil}


\begin{document}

\title{On the Birman-Krein Theorem}

\begin{abstract} It is shown that if $X$ is a unitary operator so that  a singular subspace of~$U$ is unitarily equivalent to a singular subspace of~$UX$ (or $XU$), for  each unitary operator~$U$, then $X$ is the identity operator. In other words, there is no nontrivial generalization of Birman-Krein Theorem that includes the preservation of a singular spectral subspace in this context. 
\end{abstract}

\maketitle

\section{Introduction}
We are interested in  preservations  of spectral types of unitary operators~$U$, on a  Hilbert space~$\mathcal H$,  under {\em multiplicative} (compositions, in fact) perturbations, that is, if~$X$ is another unitary operator, the perturbation has the form 
\begin{equation}\label{rightPert}
U \mapsto UX.
\end{equation} This is a {\em right} perturbation, and $U \mapsto XU$ is a {\em left} perturbation. Note that both $XU$ and $UX$ are again unitary operators. Some notation: we shall denote the identity operator by~${\mathbf 1}$ and $T_1\cong T_2$ means that the linear operators~$T_1$ and~$T_2$ are unitarily equivalent; $\mu\perp \nu$ means that the measures~$\mu$ and~$\nu$ are mutually singular; ${\mathbb T}=\left\lbrace c\in\mathbb{C}\mid |c|=1  \right\rbrace$ is the one-dimensional torus and the letter~$c$ will always denote an element of~${\mathbb T}$.

The main question to be investigated is what kind of spectral types of (any)~$U$ are preserved under such nontrivial perturbations~$X\ne {\mathbf 1}$, and this study was motivated by the corresponding question, addressed by Howland~\cite{howland1986}, in the context of additive perturbations of self-adjoint operators. Howland has concluded that there is no nontrivial generalization of Kato-Rosenblum Theorem, that is, only absolutely continuous spectra (with respect to Lebesgue measure) can be (always)  preserved under some nonzero perturbation (for instance, under trace class perturbations). Our main conclusion is similar and appear in Theorem~\ref{theorMain2}, that is, there is no generalization of Birman-Krein Theorem~\cite{birmankrein1962}, so that only the absolutely continuous spectral component of (any)~$U$ can be preserved under certain multiplicative unitary perturbation~$X\ne {\mathbf 1}$ (either right or left perturbations). 

Our arguments follow the general lines of~\cite{howland1986}, and do not go along with the Cayley transform since it does not take sums of self-adjoint operators to the multiplicative form~\eqref{eqMultAddit} of unitary ones, and since there were some nonobvious choices in the construction here, we have found the result interesting enough to justify publication. 

As an additional motivation, it is worth mentioning that multiplicative perturbations~\eqref{rightPert} of unitary operators naturally appear in quantum versions of time-periodic kicked systems~\cite{casatimolinari,combescure1990}. Let~$A$ be a self-adjoint operator describing certain quantum system  (a {\em free} one) that, at instants of time $j\tau, j\in\mathbb Z$, undergoes kicked perturbations by a self-adjoint operator~$B$; the formal operator describing such system is
\[
A \;+\;  B\,\sum_{j\in\mathbb Z} \delta(t-j\tau),
\]whose time evolution between two consecutive kicks (the Floquet operator) is given by~\cite{casatimolinari} 
\begin{equation}\label{floquetOp}
e^{-iB}\,e^{-iA\tau},
\end{equation} that is, multiplicative perturbations ($e^{-iB}$) of unitary operators ($e^{-iA\tau}$). A historically important model~\cite{casatimolinari} is the so-called {\em kicked rotator} with $A=-\mathrm d^2/\mathrm d\theta^2$ and $B=\kappa\cos\theta$, $0\le\theta<2\pi$, and a real parameter~$\kappa>0$; it is expected that, for ``most'' periods~$\tau>0$, the corresponding Floquet operator~\eqref{floquetOp} has pure point spectrum and a localized dynamics, specially for $\kappa>1$ (it is a region where classical and quantum physics would disagree), but currently, there is a mathematical proof~\cite{bourgain2002} of quantum localized dynamics  only for  small enough~$\kappa$.

A basic general discussion of multiplicative perturbations is given in Section~\ref{sectGeneral}, including a precise statement of preservation of spectral types and  of our main result, whose proof is concluded in Section~\ref{sectProof}.

\section{Multiplicative perturbations}\label{sectGeneral}
Sometimes it is convenient to write the perturbation in the form  $X=e^{iY}$, with $Y$ a bounded self-adjoint operator, so that
\begin{equation}\label{eqXemtermosY}
	X=e^{iY}=\sum_{j=0}^{\infty}\dfrac{(iY)^j}{j!}={\mathbf 1}+\sum_{j=1}^{\infty}\dfrac{(iY)^j}{j!}={\mathbf 1}+W,
\end{equation}   
with $W=\sum_{j=1}^{\infty}{(iY)^j}/{j!}$. In this way, we can write 
\begin{equation}\label{eqMultAddit}
UX=U({\mathbf 1}+W)=U+UW
\end{equation} and if $W$ is a trace class operator, then $UW$ is also trace class and, by Birman-Krein Theorem, the Lebesgue absolutely continuous subspaces of $UW$ and $U$ are unitarily equivalent for all~$U$ (the same for left perturbations $XU$).

An example is the case of multiplicative perturbations of rank one: given $\phi\in{\mathcal H}$ with  $||\phi||=1$, let $P_{\phi}$ be the projector operator onto the subspace $\mathrm{ Lin}\{ \phi\}$ generated by~$\phi$, i.e., 
\[
P_{\phi}\xi=\left\langle \phi, \xi \right\rangle \phi\,,\quad \forall \xi\in{\mathcal H},
\] which is self-adjoint and idempotent. The corresponding perturbing unitary operator, with ``intensity'' $\lambda\in\mathbf R$,  is 
 \begin{equation}\label{eqXlambda}
X_\lambda=X_{\lambda,\phi}:=e^{i\lambda P_{\phi}},
\end{equation}
 and by writing $e^{i\lambda P_{\varphi}}={\mathbf 1}+W$, it follows by~\eqref{eqXemtermosY} that 
\[
W\xi =(e^{i\lambda}-1)P_{\phi}\xi,
\] and the (right) perturbed operator has the action
\begin{equation}\label{eqU+rank1}
	U_{\lambda}:=UX_{\lambda}=U{\mathbf 1}+U(e^{i\lambda}-1)P_{\phi} ,  
\end{equation}
with $U_0=U$. To simplify statements, denote by $\mu_\psi^\lambda$ the spectral measure of the pair $(U_\lambda,\psi)$, $\psi\in{\mathcal H}$, and usually one supposes  that~$\phi$ is cyclic for $U$, that is, the closure
\[
\overline{\mathrm{ Lin}\left\lbrace U^j \phi\mid j\in \mathbb{Z}\right\rbrace }=\mathcal{H},
\] with $U^0={\mathbf 1}$. Note that if $\phi$ is cyclic for $U$, then $U^k\phi$  is also is cyclic for $U$, for all $k \in \mathbb{Z}$; furthermore, if~$\phi$ is cyclic for~$U$, then it is also cyclic for~$U_{\lambda}$ for all $\lambda\in\mathbb R$.

A particularly important result, that will be employed ahead, is the following version, for unitary operators, of the Aronszajn-Donoghue Theorem for self-adjoint ones (for proofs see Proposition 8.3 in~\cite{garcia2015} and Proposition~9.1.14 in~\cite{CMR}):

\begin{theorem}\label{thmUnitAronDono}
Let $\phi$ be cyclic for~$U$. If $\lambda_1- \lambda_2\ne 2\pi n$, for any $n\in\mathbb Z$, then the singular parts of $U_{\lambda_1}$ and~$U_{\lambda_2}$ are mutually singular, i.e., for all $0\ne\psi\in{\mathcal H}$, the singular parts of the spectral measures $\mu_\psi^{{\lambda_1}}$ and $\mu_\psi^{{\lambda_2}}$ are mutually singular.
\end{theorem}

In order to discuss more general spectral subspaces, for $\psi\in{\mathcal H}$, denote by $\mu_\psi^U$ the spectral measure of the pair $(U,\psi)$ and, given a (nonzero) finite Borel measure~$\mu$ on ${\mathbb T}$, let
\[
{\mathcal H}_\mu(U):=\left\lbrace \psi \in \mathcal{H}\mid \mu_\psi^{U}\ll \mu \right\rbrace,
\]
which is a closed subspace of the Hilbert space, whose orthogonal complement is
\[
 \left\{ \psi\in{\mathcal H} \mid \mu_\psi^U\perp \mu  \right \}.
\]
Finally, denote by $[U]_\mu$ the restriction of $U$ to~${\mathcal H}_\mu(U)$, i.e., $[U]_\mu:= U|_{{\mathcal H}_\mu(U)}$.

\begin{definition}
A unitary operator $X$  preserves~$\mu$ on the right  if $[UX]_{\mu}\cong [U]_{\mu}$, for all unitary operators~$U$.  And~$X$  preserves~$\mu$ on the left  if $[XU]_{\mu}\cong [U]_{\mu}$, for all unitary operators~$U$.
\end{definition}

\begin{remark}\label{remMany}
	\begin{enumerate}[(a)]		

\item The simple case $X=c{\mathbf 1}$ translates the spectra in~$\mathbb T$ and so does not preserve measures if~$c\ne1$. \label{remItem1}

\item As already mentioned, by Birman-Krein Theorem, if $X={\mathbf 1}+W$ with trace class~$W$, then $X$ preserves the Lebesgue measure~$\ell$ on~${\mathbb T}$.

\item If $\nu \ll \mu$, then
\begin{equation}
	[[U]_{\mu}]_{\nu}=[U]_{\nu}.
\end{equation}
\item For real $t$, define the translated (in~${\mathbb T}$) measure
\begin{equation}
	\mu_t(\cdot):=\mu\left(e^{it}\cdot \right). 
\end{equation}
Then \begin{equation}
	[U]_{\mu_t}=e^{it}[U]_{\mu}.
\end{equation}

\end{enumerate}
\end{remark}

It is enough to discuss right or left preservation of a measure, as Proposition~\ref{propPresLR} ensures. After this proposition we will just say ``$X$ preserves~$\mu$'' (as already employed above).

\begin{proposition}\label{propPresLR}
$X$  preserves $\mu$ on the right if and only if $X$  preserves~$\mu$ on the left.
\end{proposition}
\begin{proof}
Let $X^*$ denotes the adjoint of~$X$. If $X$  preserves~$\mu$ on the right, then pick $X^*U$ (which is also a unitary operator); by hypothesis $[X^*UX]_{\mu}\cong [X^*U]_{\mu}$, and since  $X$ is unitary, $[X^*UX]_{\mu} \cong [U]_{\mu}$ and one obtains that $[X^*U]\cong [U]_{\mu}$, and so~$X^*$ preserves~$\mu$ on the left. Now one has $[U]_\mu=[X^*XU]_\mu\cong [XU]_\mu$, and so~$X$ preserves~$\mu$ on the left. 

Similarly, one concludes the reciprocal.
\end{proof}

\begin{corollary}\label{corolXXstar}
$X$ preserves~$\mu$  if and only if $X^*$ preserves~$\mu$.
\end{corollary}

Theorems~\ref{theorMain2}  states the main result of this note.

\begin{theorem}\label{theorMain2}
If the unitary operator $X\ne {\mathbf 1}$ preserves~$\mu$, then $\mu\ll\ell$.
\end{theorem}

The completion of the proof of this theorem is the subject of Section~\ref{sectProof}; in the following we present some basic and useful properties. 

\begin{proposition}\label{proposition 2.3}
	Suppose that~$X$ preserves~$\mu$. Then:
	\begin{enumerate}[(a)]
		\item If $\nu \ll \mu$, then $X$ preserves $\nu$. \label{propBasica}
		\item $X$ preserves $\mu_t$, for all $t \in \mathbb{R}$. \label{propBasicb}
		\item If $X \cong Y$, then $Y$ preserves $\mu$. \label{propBasicc}
		\item \label{propBasicd} If $P_E$ is an orthogonal reducing projection  for $X$ (projection onto the closed subspace $E\subset\mathcal H$), then $P_EX$ preserves $\mu$ on $P_E(\mathcal{H}) =E$. Recall that $E^\perp$ will also be reducing for~$X$.
		\item If $Y$ also preserves $\mu$, then $XY$ preserves $\mu$. \label{propBasice}
	\end{enumerate}
\end{proposition}
\begin{proof}
	\begin{enumerate}[(a)]
		\item $ [UX]_{\nu} \cong  [  [UX]_{\mu}   ]_{\nu} \cong  [U_\mu  ]_{\nu}\cong[U]_{\nu}$.  
		\item $ [UX]_{\mu_t}=e^{it} [UX]_{\mu}= [e^{it}UX]_{\mu}$ and since $e^{it}U$ is a unitary operator, we have $ [e^{it}UX]_{\mu} \cong  [e^{it}U]_\mu=U_{\mu_t} $; hence  $ [UX]_{\mu_t}=[U]_{\mu_t}$.  
		\item If $X\cong Y$, that is, $Y=VXV^*$ for some unitary operator~$V$, then 
		\begin{eqnarray*}
		 [UY]_{\mu}= [UVXV^* ]_{\mu} &=& [ V (V^*UVXV^*V) V^*  ]_{\mu}  \\ &\cong&  [V^*UVXV^*V]_{\mu}= [V^*UVX]_{\mu}\,,
		 \end{eqnarray*} and since $V^*UV$ is unitary and $X$ preserves $\mu$, one obtains
		\[
		 [V^*UVX]_{\mu}\cong  [V^*UV]_{\mu}\cong [U]_{\mu},
		\] that is, $  [UY]_{\mu}\cong [U]_{\mu}$.
		\item Write $\mathcal{H}= E \oplus E^{\perp}$; then $X=  X |_{E} \oplus  X |_{E^{\perp}}$, which can be written as 
		\[
		X= \begin{pmatrix}
				 X |_{E} & 0\\
				0 &  X |_{E^{\perp}}  
			\end{pmatrix}.
			\] Now, if $\widetilde{U} $ and  $\widetilde{V}$ are unitary operators acting on~$E$ and~$E^{\perp}$, respectively, then
			 \[
			U=\begin{pmatrix}
				\widetilde{U} & 0\\
				0 & \widetilde{V}
			\end{pmatrix}
			\]
			 is unitary on $\mathcal{H}$. Thereby,
			\begin{equation*}
				[U]_{\mu} \cong  [UX]_{\mu}= \left[ \begin{pmatrix}
					\widetilde{U} & 0\\
					0 & \widetilde{V}
				\end{pmatrix} \begin{pmatrix}
					 X  |_{E} & 0\\
					0 &  X  |_{E^{\perp}} 
				\end{pmatrix} \right ]_{\mu} =  \begin{pmatrix}
					 [\widetilde{U} X |_{E} ]_{\mu} & 0 \\
					0 &  [\widetilde{V} X |_{E^{\perp}}   ]_{\mu}  
				\end{pmatrix} ,
			\end{equation*} and since $[U]_{\mu} = \begin{pmatrix}
				[\widetilde{U}]_{\mu} & 0 \\
				0 & [\widetilde{V}]_{\mu}
			\end{pmatrix}$, by equating the first components, one gets $ [\widetilde{U} X |_{E}   ]_{\mu} \cong \widetilde{U}_{\mu}$, and so $X |_{E}$ preserves $\mu$ on  $E=P_E ( \mathcal{H} )$.  	
		\item Indeed, if $U$ is a unitary operator, then $UX$ is also unitary and
	\begin{equation*}
		 [UXY]_{\mu} \cong  [UX]_{\mu} \cong [U]_{\mu}. 
	\end{equation*}
\end{enumerate}
\end{proof}

\section{Proof of Theorem~\ref{theorMain2}}\label{sectProof}

	 It is supposed that~$X$ preserves~$\mu$; the goal is to show that $\mu\ll\ell$. By Remark~\ref{remMany}(\ref{remItem1}), we may suppose that $X\ne c{\mathbf 1}$. By item~(\ref{propBasicc}) of  Proposition~\ref{proposition 2.3},  $R^*XR$ also preserves $\mu$, for any unitary operator~$R$. Pick the operator $R=X_{\lambda}$ from~\eqref{eqXlambda}, with $\lambda=\pi$ and some normalized vector~$\phi$ (that will be selected ahead), that is,
	\begin{equation*}
		R\xi=X_{\pi}\xi=\xi+(e^{i\pi}-1)\left\langle \phi,\xi \right\rangle \phi=\xi-2\left\langle \phi,\xi \right\rangle \phi,
	\end{equation*}
	for all $\xi \in \mathcal{H}$.
	
	By Proposition~\ref{proposition 2.3}, items~(\ref{propBasicc}) and~(\ref{propBasice}), and Corollary~\ref{corolXXstar}, the operators $R^*XR$ and  
\[
Z:=X^*R^*XR
\] preserve~$\mu$ as well. Explicitly, one has 
	\begin{align*}
		R^*XR\xi = & R^*X \left(\xi-2\left\langle\phi,\xi \right\rangle \phi    \right)\\
		= & R^* \left( X\xi-2\left\langle\phi,\xi \right\rangle X\phi     \right)\\ 
		= & R^*(X\xi)-R^*\left(-2\left\langle\phi,\xi \right\rangle X\phi  \right)\\
		= & R^*(X\xi)-2\left\langle\phi,\xi \right\rangle R^*\left(X\phi \right)\\
		= & X\xi-2\left\langle \phi,X\xi \right\rangle \phi  -  2\left\langle\phi,\xi \right\rangle \left(X\phi-2\left\langle \phi,X\phi \right\rangle \phi  \right)\\
		= & X \xi + \big(4\left\langle\phi,\xi \right\rangle \left\langle \phi,X\phi\right\rangle-2\left\langle \phi,X\xi \right\rangle     \big)\phi -2\left\langle \phi,\xi \right\rangle X\phi,   
	\end{align*}
and if $\psi=X^*\phi$,
\[
Z\xi= \xi + \left[4\left\langle\phi,\xi \right\rangle \left\langle \phi,X\phi \right\rangle-2\left\langle \phi,X\xi \right\rangle    \right]\psi-2\left\langle\phi,\xi \right\rangle \phi\,=:(\mathbf 1 + W)\xi,
\] that is,
\[
W\xi =  \left[4\left\langle\phi,\xi \right\rangle \left\langle \phi,X\phi \right\rangle-2\left\langle \phi,X\xi \right\rangle    \right]\psi-2\left\langle\phi,\xi \right\rangle \phi\,.
\]	
	
	Since $X\ne c{\mathbf 1}$, one may choose $\phi$ so that the set $\{\phi, \psi\}$ is linearly independent and so the operator~$W$  has rank precisely~$2$ (note that if $X=c\mathbf 1$, then $W=0$, as expected). Let~$e_1$ and~$ e_2$ be  two normalized and independent eigenvectors of~$W$ and note that the corresponding eigenvalues do not vanish, otherwise~$W$ would have rank smaller than~2. Observe that~$e_1$ and~$ e_2$   are also eigenvectors of the unitary operator~$Z$, hence (one may suppose that) $e_1\perp e_2$. If~$E$ is the orthogonal complement of~$e_1$, which reduces~$Z$, it follows, by Proposition~\ref{proposition 2.3}(\ref{propBasicd}), that the operator $Z_2:= Z|_{E}$ also preserves~$\mu$ on~$E$, and since we now have a nonzero rank one perturbation, $Z_2$ has the form~\eqref{eqXlambda}, that is, 
\[
Z_2 = \mathbf 1 + (e^{i\lambda}-1)P_{ e_2},
\] for some $\lambda\ne 2\pi k, k\in\mathbb Z$.
	
	Finally, pick a unitary operator~$\dot U$ on~$E$ with~$ e_2$ one of its cyclic vectors; hence
\[
\dot UZ_2 = \dot U + \dot U(e^{i\lambda}-1)P_{ e_2}	
\] has the form~\eqref{eqU+rank1}. Since $Z_2$ preserves~$\mu$,
\[
[\dot UZ_2]_\mu\cong [\dot U]_\mu
\] and, by Theorem~\ref{thmUnitAronDono},  if~$\mu$ has a (nonzero) singular component, $Z_2$ could not preserve~$\mu$, and one concludes that~$\mu$ is absolutely continuous with respect to Lebesgue measure. 
	
\subsection*{Acknowledgments.} CRdO thanks the partial support by CNPq (under contract number 303689/2021-8), and PAD thanks the support by CAPES (Brazilian agencies).

\end{document}